\documentclass[11pt]{article}

\usepackage{amsmath,amssymb,amsthm,fullpage}
\usepackage[sort]{cite}
\newtheorem{theorem}{Theorem}

\newtheorem{lemma}[theorem]{Lemma}
\newcommand{\maxdegree}{D}
\begin{document}

\title{Asymmetric Palette Sparsification, Slightly Simplified}
\date{}
\author{Andrew McGregor\thanks{This author is supported by the National Science Foundation under grant CCF-2521579.}}
\maketitle

\begin{abstract}
We present a slightly simplified analysis of the asymmetric palette sparsification result by Assadi and Yazdanyar [TheoretiCS, 2026]. The motivation is mainly pedagogical; our approach avoids hypergeometric concentration bounds and extra constant factors in the palette size.  
\end{abstract}

\section{Introduction}

Let $G=(V,E)$ be an $n$-vertex graph of maximum degree $D$.
In the \emph{palette sparsification problem}, each vertex $v\in V$ independently samples
a random ``palette'' $
S(v) \subseteq [\maxdegree+1]$.
The goal is to choose the palettes so that, with high probability over the random
sampling, there exists a proper coloring where the color assigned to  $v\in V$ is in the set $S(v)$.
The main palette sparsification theorem of Assadi, Chen, and Khanna~\cite{AssadiChenKhanna2019} states that
there is an absolute constant $C>0$ such that if every vertex samples
$C\log n$
colors independently and uniformly from $[\maxdegree+1]$, then with high probability
the graph $G$ admits a proper $(\maxdegree+1)$-coloring using only the sampled
palettes.

More recently, Assadi and Yazdanyar~\cite{AssadiYazdanyar2026} introduced \emph{asymmetric palette sparsification},
in which the palette sizes are allowed to vary from vertex to vertex.
Their result shows that one can guarantee a proper coloring even when only the
\emph{average} palette size is small, at the cost of a weaker bound than the
earlier theorem: whereas the original palette sparsification theorem gives
$|S(v)|=O(\log n)$ for every vertex, the asymmetric version achieves an
$O(\log^2 n)$ bound on the average value of $|S(v)|$.
Although asymmetric palette sparsification gives a weaker guarantee, it has two compensating advantages. Its proof is substantially simpler, and the resulting coloring can be recovered by an especially simple polynomial-time greedy procedure.
In this note, we (slightly) simplify even further.
The motivation is mainly pedagogical; our approach avoids hypergeometric concentration bounds and extra constant factors in the palette size.  Specifically, we prove:

\begin{theorem}\label{thm:random-palette-coloring}
Let $G=(V,E)$ be a graph on $n$ vertices and maximum degree  $D$. For each vertex $v\in V$, choose $t(v)\sim \mathrm{Unif}[0,1]$ independently, and then form a random set $S(v)\subseteq [D+1]$ by including each color independently with probability
\[
p(v)=\min \left\{1,\frac{\log (n/\delta)}{1+t(v)D}\right\}.
\]
Process the vertices in decreasing order of the values $t(v)$, and when vertex $v$ is processed, assign it any color in $S(v)$ that differs from the colors already assigned to its previously processed neighbors.
This procedure succeeds in coloring every vertex with probability at least $1-\delta$ and the expected palette size is at most $\frac{D+1}{D} \cdot \log (n/\delta)\log(D+1)$.
\end{theorem}

\section{Proof of Theorem}

\begin{lemma}\label{lem:single-vertex}
Suppose each neighbor of a vertex $v$ is already colored independently with probability $1-\tau$, where $\tau\in[0,1]$, and the colors used come from $[D+1]$. Let $S\subseteq [D+1]$ be formed by including each color independently with probability
\[
p=\min \left\{1,\frac{\log (1/\delta)}{1+\tau D}\right\}.
\]
Then the probability that $S$ contains no color available for $v$ is at most $\delta$.
\end{lemma}

\begin{proof}
Let the number of neighbors of $v$ that are not yet colored be $Y\sim \mathrm{Bin}(d,\tau)$ where 
$d\leq D$ is the degree of $v$. Hence the number of already-colored neighbors is $d-Y$, so at most $d-Y$ distinct colors are forbidden to $v$. Therefore the number $A$ of colors available for $v$ satisfies
\[
A\ge (D+1)-(d-Y)=D+1-d+Y.
\]

Conditioned on $Y=y$, the event that $S$ contains no available color has probability at most
\[
(1-p)^A\le (1-p)^{D+1-d+y}.
\]
The lemma is trivial  when $p=1$ since then $|S|=D+1$ is greater than the degree of $v$. Henceforth, assume $p=\log(1/\delta)/(1+\tau D)$. Taking the expectation over $Y$,
\begin{eqnarray*}
\Pr[S\text{ misses all available colors}]
   & \le & (1-p)^{D+1-d}\,\mathbb{E}\bigl[(1-p)^Y\bigr] \\
   &=&  (1-p)^{D+1-d}(1-\tau p)^d \\ 
   &\leq & \exp (-p(D+1-d)-\tau p d ) \\
   &\leq & \exp (-p(1+\tau D)) = \delta \ .    \end{eqnarray*}
   where the equality follows because $Y$ is the sum of independent indicator functions and $\mathbb{E}[(1-p)^X]=1-\mathbb{E}[X]p$ for any indicator  function $X$. The second inequality uses $1-x\leq e^{-x}$ for all $x$.
\end{proof}

\begin{lemma}[Palette Size]
The expected palette size is \[(D+1)\mathbb{E}[p(v)] \leq \log (n/\delta)\log(D+1) \frac{D+1}{D} \ .\]
\end{lemma}
\begin{proof}
$\mathbb{E}[p(v)]$ can be evaluated exactly but to prove the lemma, we just note 
$
p(t)\le \frac{\log (n/\delta)}{1+tD},
$
and so: \[
\mathbb{E}[p(t)]
   \le \int_0^1 \frac{\log (n/\delta)}{1+tD}\,dt
   = \frac{\log (n/\delta)}{D}\int_1^{D+1}\frac{1}{x}\,dx
   = \frac{\log (n/\delta)}{D}\log(D+1) \ . \qedhere \]
\end{proof}

\begin{proof}[Proof of Theorem \ref{thm:random-palette-coloring}]
Fix a vertex $v$, and condition on the value $t(v)=\tau$. Since the random variables $\{t(u):u\neq v\}$ are independent and uniform in $[0,1]$, each neighbor $u$ of $v$ satisfies
$
\Pr[t(u)<\tau\mid t(v)=\tau]=\tau$
independently of the others. Thus, conditioned on $t(v)=\tau$, the number of neighbors of $v$ that have not yet been processed is distributed as $\mathrm{Bin}(\deg(v),\tau)$. Applying Lemma~\ref{lem:single-vertex} with $d=\deg(v)$ and this value of $\tau$, we obtain
$
\Pr[v\text{ cannot be colored}\mid t(v)=\tau]\le \delta/n$. Averaging over $\tau$,
\[
\Pr[v\text{ cannot be colored}]\le \delta/n.
\]
Taking the union bound over all vertices, we deduce $
\Pr[\text{some vertex cannot be colored}]
   \le \delta$.
\end{proof}

\paragraph{Acknowledgements and Developments.} The author had the idea that sidestepped concentration bounds (because he didn't want to cover hypergeometric concentration bounds in a class). ChatGPT 5.4 was used to optimize the distribution of $p(\cdot)$ to avoid extra constants. Thanks to Amit Chakrabarti for feedback on the manuscript. Lastly, we note that El-Hayek, Henzinger, and Zheng \cite{elhayek2026llmsusedsimplifyalgorithms} just announced an independent simplification of the Assadi and Yazdanyar result.

{\small
\bibliographystyle{plain}
\bibliography{palette}
}
\end{document}